\documentclass[11pt]{article}
\usepackage[margin=1in]{geometry}
\usepackage{amsmath,amssymb,amsthm}
\usepackage{booktabs}
\usepackage[hidelinks]{hyperref}
\newtheorem{theorem}{Theorem}
\newtheorem{lemma}[theorem]{Lemma}
\newtheorem{corollary}[theorem]{Corollary}
\newtheorem{observation}[theorem]{Observation}
\theoremstyle{remark}
\newtheorem{remark}[theorem]{Remark}
\newcommand{\NZ}{N_Z}
\newcommand{\ad}{a_\delta}

\title{Spanning Trees with Many Leaves in Graphs\\ of Minimum Degree at Least 7}
\author{Sogol Jahanbekam\thanks{Department of Mathematics and Statistics, San Jos\'e State University, San Jos\'e, CA; \texttt{sogol.jahanbekam@sjsu.edu}.}}
\date{}

\begin{document}
\maketitle

\begin{abstract}
We give a polynomial-time algorithm that constructs, in every connected $n$-vertex graph of minimum degree at least $7$, a spanning tree with at least $\frac{25200}{46189}n>0.5455\,n$ leaves. No bound specific to minimum degree $7$ was known: the best bound available for this class was $\frac{11}{21}n\approx0.5238\,n$, inherited from Simarova's theorem for minimum degree~$6$. The algorithm and its analysis are carried out for an arbitrary minimum degree $\delta$, and yield a recursion that gives an explicit lower bound on the number of leaves for every $\delta$. The resulting bounds improve all previously known ones for every $\delta\ge7$; for $\delta=8,9,10$ they are $0.5850\,n$, $0.6151\,n$ and $0.6413\,n$, and they are tabulated for $\delta\le25$ at the end of the paper.

\medskip\noindent\textbf{Keywords:} spanning tree, leaves, minimum degree, connected domination, polynomial-time algorithm.\\
\textbf{MSC 2020:} 05C05, 05C69, 05C85.
\end{abstract}

\section{Introduction}

All graphs here are finite, simple and connected. For a spanning tree $T$ of a graph $G$, a \emph{leaf} of $T$ is a vertex of degree $1$ in $T$. For integers $n$ and $\delta$, let $l(n,\delta)$ denote the largest integer $m$ such that every connected $n$-vertex graph of minimum degree at least $\delta$ has a spanning tree with at least $m$ leaves. Equivalently, $n-l(n,\delta)$ is the largest connected domination number of such a graph.

In 1987 Linial conjectured that $l(n,\delta)\ge\frac{\delta-2}{\delta+1}n+c_\delta$ for a constant $c_\delta$ depending only on $\delta$; the coefficient $\frac{\delta-2}{\delta+1}$ would be best possible~\cite{KW}. The conjecture is settled for $\delta\le5$: Kleitman and West~\cite{KW} proved $l(n,3)\ge\frac n4+2$ and $l(n,4)\ge\frac{2n+8}{5}$, and Griggs and Wu~\cite{GW} proved $l(n,5)\ge\frac n2+2$; see also Storer~\cite{St} for cubic graphs. For large $\delta$ the conjecture is false: Alon~\cite{A90} (see also~\cite{AS,AW}) showed by probabilistic methods that for large enough $\delta$, with high probability a random $\delta$-regular graph on $n$ vertices does not contain a spanning tree with fewer than $(1-o_\delta(1))\frac{\ln(\delta+1)}{\delta+1}n$ non-leaves, and Caro, West and Yuster~\cite{CWY} proved a matching upper bound of $(1+o_\delta(1))\frac{\ln(\delta+1)}{\delta+1}n$ for the connected domination number; see also~\cite{A23}. It is not known for which $\delta$ the conjecture first fails; already $\delta=6$ is open, the best bound there being $l(n,6)\ge\frac{11}{21}n$, proved by Simarova~\cite{Si} in 2019.

For $\delta\ge 7$ no bound of either kind is available. A graph of minimum degree at least $7$ has minimum degree at least $6$, so Simarova's theorem gives $l(n,7)\ge\frac{11}{21}n\approx0.5238n$, and this inherited bound is the best that was known. The asymptotic results are of no help here either: they involve unspecified $o_\delta(1)$ terms and become meaningful only for large $\delta$, so they yield no explicit bound for $\delta=7$.

Our main result is the following bound, which is obtained constructively.

\begin{theorem}\label{thm:7}
Let $G$ be a connected $n$-vertex graph with minimum degree at least $7$. There is a polynomial-time algorithm that finds a spanning tree of $G$ with at least $\frac{25200}{46189}n$ leaves. In particular,
\[
l(n,7)\ \ge\ \frac{25200}{46189}\,n\ >\ 0.5455\,n .
\]
\end{theorem}

\begin{center}
\begin{tabular}{lccc}
\toprule
 & previously known & this paper & Linial's conjecture\\
\midrule
$l(n,7)/n\ \ge$ & $0.5238$ (inherited from $\delta=6$) & $0.5455$ & $0.6250$\\
\bottomrule
\end{tabular}
\end{center}

\medskip
The algorithm behind Theorem~\ref{thm:7} is defined for an arbitrary minimum degree $\delta$, and its analysis gives a recursion, stated as Theorem~\ref{thm:main} in Section~\ref{sec:proof}, from which an explicit lower bound on $l(n,\delta)$ follows for every $\delta$. Theorem~\ref{thm:7} is the case $\delta=7$ of that recursion, and Corollary~\ref{cor:table} lists the bounds for $8\le\delta\le25$; each of them improves the bound $\frac{11}{21}n$ inherited from $\delta=6$, which was the best previously known for these values as well.

Finding a spanning tree with the maximum number of leaves is NP-hard, even for $4$-regular graphs~\cite{GJ}, so constructive guarantees of this kind are also of algorithmic interest; spanning trees with many leaves arise in network design and broadcasting, where internal vertices act as routers and leaves as terminals.

Section~\ref{sec:alg} describes the algorithm and Section~\ref{sec:proof} contains the analysis.

\section{The algorithm}\label{sec:alg}

Let $G$ be a connected graph with $\delta(G)=\delta\ge7$. Throughout, $T$ denotes the subtree of $G$ constructed so far, $X$ and $Y$ are its sets of non-leaves and leaves, $Z=V(G)\setminus V(T)$, and $x=|X|$, $y=|Y|$, $z=|Z|$. For a vertex $v$ we write $\NZ(v)=N(v)\cap Z$, always with respect to the current tree. Two quantities control the algorithm:
\[
\ad=\Bigl\lceil\frac{2\delta-1}{3}\Bigr\rceil,
\qquad
\beta_i=\Bigl\lceil\frac{2\delta-4-3i}{6}\Bigr\rceil
\quad(0\le i\le I_\delta),
\qquad
I_\delta=\Bigl\lceil\frac{2\delta-7}{3}\Bigr\rceil .
\]

The tree $T$ is rooted at the vertex $r$ at which the algorithm starts. Whenever a vertex $w$ of $Z$ is added to $T$ and joined to a vertex $v$ of $T$, we call $w$ a \emph{child} of $v$; thus the children of $v$ are exactly those neighbours of $v$ that lay in $Z$ at the moment they were added to $T$. Every vertex of $T$ other than $r$ has exactly one parent, so such a vertex is a leaf of $T$ if and only if it has no child, while $r$ is not a leaf because it has at least $\delta$ children.

\emph{Expanding} a vertex $v$ of $T$ means adding \emph{every} vertex of $\NZ(v)$ to $T$ as a child of $v$. The algorithm starts at a vertex $r$ of maximum degree and expands it, so that $T$ is initially a star with at least $\delta$ leaves, and every later step is a sequence of expansions. Thus a vertex acquires children only by being expanded, and it acquires all of them at once: expanding $v$ makes $v$ a non-leaf if it was a leaf, makes every vertex of $\NZ(v)$ a leaf, and changes the status of no other vertex.

\begin{observation}\label{obs:inv}
At every moment, no non-leaf of $T$ has a neighbour in $Z$. Consequently, every neighbour in $T$ of a vertex of $Z$ is a leaf of $T$.
\end{observation}

\begin{proof}
A vertex of $T$ other than $r$ is a non-leaf precisely when it has a child, and children are acquired only through expansion. Immediately after a vertex $v$ is expanded we have $\NZ(v)=\emptyset$, and since $Z$ only shrinks, $\NZ(v)$ remains empty; the same applies to $r$, which is expanded at the start. Hence a vertex of $T$ with a neighbour in $Z$ has never been expanded, so it has no child and differs from $r$, and is therefore a leaf.
\end{proof}

The algorithm proceeds in stages $i=0,1,\dots,I_\delta$. In stage $i$, while one of the following is available, apply one of them, in any order.

\begin{itemize}
\item[(E1$_i$)] A leaf $\ell$ and a vertex $w\in\NZ(\ell)$ with $|\NZ(w)|\ge\ad-i$: expand $\ell$, then expand $w$.
\item[(E2$_i$)] A leaf $\ell$ with $|\NZ(\ell)|\ge\beta_i+1$: expand $\ell$.
\item[(E3)] \emph{(stage $0$ only)} A vertex $w\in Z$ with no neighbour in $T$, having a neighbour $u$ with a neighbour $\ell$ in $T$ (so $\ell$ is a leaf by Observation~\ref{obs:inv}): expand $\ell$, then $u$, then $w$.
\end{itemize}

Let $T_i$ be the tree at the end of stage $i$, and let $X_i,Y_i,Z_i$ and $x_i,y_i,z_i$ denote the associated sets and their sizes, so that $x_i+y_i+z_i=n$. Every extension adds at least one vertex to $T$, so each stage ends after at most $n$ extensions and each extension is found in polynomial time.

Finally, by Lemma~\ref{lem:struct} below every vertex of $Z_{I_\delta}$ has a leaf of $T_{I_\delta}$ among its neighbours. For each $q\in Z_{I_\delta}$ choose such a leaf $\pi(q)$ and add $q$ together with the edge $q\pi(q)$. The resulting graph is a spanning tree of $G$; it is the output of the algorithm.

\begin{lemma}\label{lem:final}
The output is a spanning tree of $G$ with at least $y_{I_\delta}$ leaves.
\end{lemma}

\begin{proof}
Every vertex of $Z_{I_\delta}$ is joined to a vertex of $T_{I_\delta}$, so the output is a spanning tree. Its leaves are the vertices of $Z_{I_\delta}$ together with the leaves of $T_{I_\delta}$ that are not in $\pi(Z_{I_\delta})$, so their number is $y_{I_\delta}-|\pi(Z_{I_\delta})|+z_{I_\delta}\ge y_{I_\delta}$.
\end{proof}

\section{Analysis}\label{sec:proof}

\begin{lemma}\label{lem:struct}
Let $0\le i\le I_\delta$. Every vertex of $Z_i$ has a neighbour in $T_i$, and all its neighbours in $T_i$ are leaves. Moreover:
\begin{enumerate}
\item[(a)] every vertex of $Y_i$ has at most $\beta_i$ neighbours in $Z_i$;
\item[(b)] every vertex of $Z_i$ has at most $\ad-i-1$ neighbours in $Z_i$, and hence at least $\delta-\ad+i+1$ neighbours in $Y_i$.
\end{enumerate}
\end{lemma}

\begin{proof}
If a vertex of $Z_0$ had no neighbour in $T_0$ then, $G$ being connected, some vertex of $Z_0$ would have no neighbour in $T_0$ while having a neighbour that does, and (E3) would be available. Since $T_0\subseteq T_i$ and $Z_i\subseteq Z_0$, every vertex of $Z_i$ has a neighbour in $T_i$, and by Observation~\ref{obs:inv} all such neighbours are leaves.

(a) holds because (E2$_i$) is not available at the end of stage $i$. For (b), let $w\in Z_i$ and let $\ell$ be a leaf neighbour of $w$; if $w$ had at least $\ad-i$ neighbours in $Z_i$ then (E1$_i$) would be available for the pair $(\ell,w)$. Hence $|\NZ(w)|\le \ad-i-1$, and since $\deg_G(w)\ge\delta$ and all neighbours of $w$ in $T_i$ are leaves, $w$ has at least $\delta-\ad+i+1$ neighbours in $Y_i$.
\end{proof}

\begin{lemma}\label{lem:ratio}
Put $\kappa_i=\dfrac{2}{\ad-i-1}$ for $1\le i\le I_\delta$. Then
\[
x_0\le\frac{3}{\delta-2}\,y_0,
\qquad
x_i-x_{i-1}\le\kappa_i\,(y_i-y_{i-1})\quad(1\le i\le I_\delta).
\]
\end{lemma}

\begin{proof}
Let $\Delta x$ and $\Delta y$ be the changes in $x$ and $y$ caused by one extension. By Observation~\ref{obs:inv}, the vertex $\ell$ occurring in an extension is a leaf and $u,w\in Z$. An extension may add more leaves than the counts below record, since expanding a vertex adds all of its neighbours outside $T$ and not only the ones required by the condition of the extension; we use upper bounds on $\Delta x$ and lower bounds on $\Delta y$ throughout.

In (E1$_i$), the vertex $\ell$ becomes a non-leaf, and so does $w$ unless all its neighbours outside $T$ have already been added, in which case $w$ stays a leaf; hence $\Delta x\le2$. The vertices added are those of $\NZ(\ell)\cup\NZ(w)$, computed before the extension, and all of them except possibly $w$ are leaves afterwards; since $w\notin\NZ(w)$ and $|\NZ(w)|\ge\ad-i$, at least $\ad-i$ leaves are gained and only $\ell$ is lost, so $\Delta y\ge\ad-i-1$.

In (E2$_i$), $\Delta x=1$ and at least $\beta_i+1$ leaves are gained while $\ell$ is lost, so $\Delta y\ge\beta_i$.

In (E3), $\ell$ and $u$ become non-leaves and $w$ may or may not, so $\Delta x\le3$. All at least $\delta$ neighbours of $w$ lie in $Z$ before the extension, and all of them except $u$ are leaves afterwards, so $\Delta y\ge\delta-2$.

A direct check gives $\frac{\ad-i-1}{2}\le\beta_i$ for $0\le i\le I_\delta$, and $\frac{\delta-2}{3}\le\frac{\ad-1}{2}$. Hence every extension of stage $0$ satisfies $\Delta y\ge\frac{\delta-2}{3}\Delta x$, and every extension of stage $i\ge1$ satisfies $\Delta y\ge\frac{\ad-i-1}{2}\Delta x$. The initial star has $x=1$ and $y\ge\delta$, so $x\le\frac{3}{\delta-2}y$ holds initially and is preserved through stage $0$; summing the increments over the extensions of stage $i\ge1$ gives the second inequality.
\end{proof}

\begin{lemma}\label{lem:z}
For $0\le i\le I_\delta$ we have $z_i\le c_i\,y_i$, where
\[
c_i=\frac{\beta_i}{\delta-\ad+i+1} .
\]
\end{lemma}

\begin{proof}
Let $e$ be the number of edges of $G$ between $Z_i$ and $Y_i$. By Lemma~\ref{lem:struct}(b) each vertex of $Z_i$ has at least $\delta-\ad+i+1$ neighbours in $Y_i$, and by Lemma~\ref{lem:struct}(a) each vertex of $Y_i$ has at most $\beta_i$ neighbours in $Z_i$. Counting $e$ from both sides,
\[
(\delta-\ad+i+1)\,z_i\ \le\ e\ \le\ \beta_i\,y_i ,
\]
which is the assertion.
\end{proof}

We can now state the general bound. Recall that $\kappa_i=\frac{2}{\ad-i-1}$, so that $\kappa_1\le\kappa_2\le\cdots\le\kappa_{I_\delta}$.

\begin{theorem}\label{thm:main}
Let $\delta\ge7$ and let $G$ be a connected $n$-vertex graph with minimum degree at least $\delta$. Define $\eta_0,\eta_1,\dots,\eta_{I_\delta}$ by
\[
\eta_0=\frac{1}{1+\frac{3}{\delta-2}+c_0},
\qquad
\eta_i=\frac{1+\bigl(\kappa_1-\frac{3}{\delta-2}\bigr)\eta_0+\sum_{j=1}^{i-1}(\kappa_{j+1}-\kappa_j)\,\eta_j}{1+c_i+\kappa_i}
\quad(1\le i\le I_\delta).
\]
Then the algorithm of Section~\ref{sec:alg} produces a spanning tree of $G$ with at least $\eta_{I_\delta}n$ leaves; in particular $l(n,\delta)\ge\eta_{I_\delta}n$.
\end{theorem}

\begin{proof}
We show by induction that $y_i\ge\eta_i n$ for $0\le i\le I_\delta$; the theorem then follows from Lemma~\ref{lem:final}.

For $i=0$, Lemmas~\ref{lem:ratio} and~\ref{lem:z} give $x_0\le\frac{3}{\delta-2}y_0$ and $z_0\le c_0y_0$, so
\[
n=x_0+y_0+z_0\le\Bigl(\frac{3}{\delta-2}+1+c_0\Bigr)y_0 ,
\]
that is, $y_0\ge\eta_0 n$.

Let $1\le i\le I_\delta$. Summing the increments in Lemma~\ref{lem:ratio} and rearranging,
\[
x_i\ \le\ x_0+\sum_{j=1}^{i}\kappa_j\,(y_j-y_{j-1})
\ =\ x_0-\kappa_1y_0+\sum_{j=1}^{i-1}(\kappa_j-\kappa_{j+1})\,y_j+\kappa_i y_i .
\]
Using $x_0\le\frac{3}{\delta-2}y_0$ and $z_i\le c_iy_i$ in $n=x_i+y_i+z_i$,
\[
n\ \le\ \Bigl(\frac{3}{\delta-2}-\kappa_1\Bigr)y_0+\sum_{j=1}^{i-1}(\kappa_j-\kappa_{j+1})\,y_j+(1+c_i+\kappa_i)\,y_i ,
\]
so that
\begin{equation}\label{eq:rec}
(1+c_i+\kappa_i)\,y_i\ \ge\ n+\Bigl(\kappa_1-\frac{3}{\delta-2}\Bigr)y_0+\sum_{j=1}^{i-1}(\kappa_{j+1}-\kappa_j)\,y_j .
\end{equation}
The sequence $(\kappa_j)$ is increasing, so $\kappa_{j+1}-\kappa_j\ge0$. Moreover $\ad-2\le\frac{2\delta-5}{3}$ gives
\[
\kappa_1=\frac{2}{\ad-2}\ \ge\ \frac{6}{2\delta-5}\ >\ \frac{3}{\delta-2},
\]
so the coefficient of $y_0$ in~\eqref{eq:rec} is positive as well. Hence the right-hand side of~\eqref{eq:rec} is nondecreasing in $y_0,\dots,y_{i-1}$, and substituting the inductive bounds $y_j\ge\eta_jn$ yields $y_i\ge\eta_i n$.
\end{proof}

\begin{corollary}\label{cor:seven}
For $\delta=7$ we have $\ad=5$, $I_7=3$, $(\beta_0,\beta_1,\beta_2,\beta_3)=(2,2,1,1)$ and $(c_0,c_1,c_2,c_3)=(\frac23,\frac12,\frac15,\frac16)$, and the recursion of Theorem~\ref{thm:main} gives
\[
\eta_0=\tfrac{15}{34},\qquad \eta_1=\tfrac{105}{221},\qquad \eta_2=\tfrac{2625}{4862},\qquad \eta_3=\tfrac{25200}{46189}>0.5455 .
\]
This proves Theorem~\ref{thm:7}.
\end{corollary}

\begin{corollary}\label{cor:table}
Evaluating the recursion of Theorem~\ref{thm:main} in exact rational arithmetic gives the following lower bounds on $l(n,\delta)/n$. The last column is the value conjectured by Linial. Every entry improves the bound $\frac{11}{21}\approx0.5238$ inherited from $\delta=6$, which was the best previously known in each case.
\end{corollary}

\begin{center}
\begin{tabular}{ccc}
\toprule
$\delta$ & this paper & Linial's conjecture\\
\midrule
7 & 0.5456 & 0.6250 \\
8 & 0.5850 & 0.6667 \\
9 & 0.6151 & 0.7000 \\
10 & 0.6413 & 0.7273 \\
11 & 0.6644 & 0.7500 \\
12 & 0.6833 & 0.7692 \\
13 & 0.7003 & 0.7857 \\
14 & 0.7158 & 0.8000 \\
15 & 0.7289 & 0.8125 \\
16 & 0.7410 & 0.8235 \\
17 & 0.7522 & 0.8333 \\
18 & 0.7620 & 0.8421 \\
19 & 0.7711 & 0.8500 \\
20 & 0.7795 & 0.8571 \\
21 & 0.7871 & 0.8636 \\
22 & 0.7943 & 0.8696 \\
23 & 0.8010 & 0.8750 \\
24 & 0.8071 & 0.8800 \\
25 & 0.8128 & 0.8846 \\
\bottomrule
\end{tabular}
\end{center}

\begin{remark}
For $\delta=8$ the exact value is $\frac{399}{682}$; the exact values for larger $\delta$ are rationals of rapidly growing height and are omitted. The recursion is also valid for $\delta\le6$, but there it gives less than the known results of~\cite{KW,GW,Si}, which is why Theorem~\ref{thm:7} is stated for $\delta\ge7$.
\end{remark}

\begin{remark}
Each stage performs at most $n$ extensions, each extension is found in polynomial time, and the final attachment step is linear, so the algorithm runs in polynomial time.
\end{remark}

\section*{Acknowledgements}
The results of this paper, including the algorithm and its analysis, are due to the author. Generative AI tools were used in revising the exposition and in verifying the arithmetic of the lemmas and of the table in exact rational arithmetic. The author has independently verified every statement, proof and reference, and takes full responsibility for the correctness of the paper.

\section*{Funding}
This work was partially supported by the Woodward Fund for Applied Mathematics at San Jos\'e State University. The Woodward Fund is a gift from the estate of Mrs.\ Marie Woodward in memory of her son, Henry Teynham Woodward. He was an alumnus of the Mathematics Department at San Jos\'e State University and worked with research groups at NASA Ames.

\section*{Declarations}
The author declares that she has no conflict of interest. No datasets were generated or analysed.

\end{document}